\documentclass[runningheads]{llncs}
\usepackage[utf8]{inputenc}
\usepackage{amsmath}
\usepackage{amssymb}
\usepackage{multirow}
\usepackage{graphicx}
\usepackage{array}
\usepackage{float}
\usepackage{wrapfig}
\usepackage{subcaption}
\usepackage{color}
\usepackage{xspace}
\usepackage{todonotes}
\usepackage{enumerate}
\usepackage{verbatim}
\usepackage{algorithm,algpseudocode}

\definecolor{darkgreen}{RGB}{0,100,0}
\definecolor{orange}{RGB}{200,100,0}

\ifdefined\DEBUG{}

\newcommand{\mes}[1]{\textcolor{orange}{#1}}
\newcommand{\sumi}[1]{\textcolor{red}{#1}}

\else

\newcommand{\mes}[1]{#1}
\newcommand{\sumi}[1]{#1}

\fi

\newcommand{\OO}{\mathcal{O}}

\newcommand{\eps}{\varepsilon}
\newcommand{\ie}{\textit{i.e.}}
\newcommand{\eg}{\textit{e.g.}}
\newcommand{\opt}{\mathrm{OPT}}
\newcommand{\bset}{\mathcal{B}}
\newcommand{\lleaf}{{\rm lleaf}}
\newcommand{\rleaf}{{\rm rleaf}}

\newtheorem{observation}{Observation}
\usepackage{authblk}

\begin{document}

\title{PTAS for Steiner Tree on Map Graphs\thanks{The first three authors were supported by the NCN grant number 2015/18/E/ST6/00456}}

\author{Jarosław Byrka\inst{1}\orcidID{0000-0002-3387-0913} \and
Mateusz Lewandowski\inst{1}\orcidID{0000-0003-2912-099X} \and
Syed Mohammad Meesum\inst{1}\orcidID{0000-0002-1771-403X} \and
Joachim Spoerhase\inst{2}\orcidID{0000-0002-2601-6452} \and
Sumedha Uniyal\inst{2}\orcidID{0000-0002-3999-7827}}
\institute{Institute of Computer Science, University of Wrocław, Poland
\and
Aalto University, Espoo, Finland
}

\authorrunning{J. Byrka et al.}

\maketitle

\begin{abstract}

We study the Steiner tree problem on map graphs, which substantially generalize planar graphs as they allow arbitrarily large cliques. We obtain a PTAS for Steiner tree on map graphs, which builds on the result for planar edge weighted instances of Borradaile et al.  

The Steiner tree problem on map graphs can be casted as a special case of the planar node-weighted Steiner tree problem, for which only a $2.4$-approximation is known. We prove and use a contraction decomposition theorem for planar node weighted instances.
This readily reduces the problem of finding a PTAS for planar node-weighted Steiner tree to finding a spanner, $\ie$, a constant-factor approximation containing a nearly optimum solution.
Finally, we pin-point places where known techniques for constructing such spanner fail on node weighted instances and further progress requires new ideas.

\end{abstract}


\section{Introduction}
The Steiner tree problem has been recognized by both theorists and practitioners as one of the most fundamental problems in combinatorial optimization and network design. In this classical NP-hard problem, given a graph $G=(V,E)$ and a set of terminals $R$ the goal is to find a tree connecting all the terminals of minimum cost. The long sequence of papers established the current best approximation ratio of $1.386$~\cite{BGRS13}.

The node-weighted Steiner tree problem (\textsc{nwst}) is a generalization of the above problem. This can be easily seen by placing additional vertices in the middle of edges. Moreover, an easy reduction shows that this variant is as difficult to approximate as the Set Cover problem. Indeed, there are greedy $O\left(\log n\right)$ approximation algorithms \cite{KR95,GMNS99} matching this lower bound.

Much research has been devoted to studying combinatorial optimization problems on planar graphs, $\ie$ graphs that can be drawn on a plane without crossings. This natural restriction allows for better results, especially in terms of approximation algorithms. To this end, multiple techniques have been developed using the structural properties of planar graphs, including balanced separators \cite{LT80,AGKKW98,BFH19}, bidimensionality \cite{DH08}, local search \cite{CG15,CKM19}, shifting technique \cite{Baker94}. Such techniques are immediately applicable to a wide range of problems.

The Steiner problems, however, require more involved construction. The already established framework for approximation schemes for Steiner problems on planar graphs is briefly as follows:
\begin{enumerate}
  \item \textbf{Construct a spanner}

    Spanner is a subgraph of the input graph satisfying two properties: (1) total cost of the spanner is at most $f(\epsilon)$ times the cost of the optimum solution and (2) the spanner preserves nearly-optimum solution. Planarity of the input graph is heavily used to find such spanner.
  \item \textbf{Apply contraction decomposition theorem}

    The edges of the spanner are partitioned into $k$ sets, such that contracting each set results in a graph of constant treewidth. Because we started with a cheap spanner, \mes{there is a choice of $k$ for which} the cheapest such set of edges has cost $\eps \cdot \opt$. This partitioning is given by a contraction decomposition theorem~\cite{Klein08} (also known as \textit{thinning}) which can be obtained by applying the Baker's shifting technique~\cite{Baker94} to the dual graph.
  \item \textbf{Solve bounded-treewidth instances}

    The remaining instance is solved exactly (or in some cases approximately) in polynomial time via dynamic programming.
\end{enumerate}

Indeed, the PTAS construction for the Steiner tree problem due to Borradaile, Klein and Mathieu~\cite{bkm09} uses exactly this framework. The follow-up results for other problems like Steiner forest~\cite{BHM11,EKM12}, prize-collecting Steiner tree~\cite{BCEHKM11}, group Steiner tree~\cite{BDHM16} successfully follow the same approach (although adding new important ingredients like spanner bootstrapping).

On the other hand --- despite many efforts --- the status of the node-weighted Steiner tree problem is not yet decided on planar graphs. The state-of-the-art algorithms achieve only constant factor approximations. A GW-like primal-dual method gives a ratio of $6$~\cite{DHK14}, which was further simplified and improved to 3 by Moldenhauer~\cite{Moldenhauer13}.
The current best result is a more involved $2.4$-approximation by Berman and Yaroslavtsev~\cite{BermanY12}. However, as the integrality gap of the LP used by the above primal-dual algorithms is lower-bounded by $2$, such approach does not appear to lead to an approximation scheme.

\subsection{Motivation for Map graphs}

\begin{figure}[t]
\centering
\begin{subfigure}{.33\textwidth}
  \centering
  \includegraphics[width=.9\linewidth]{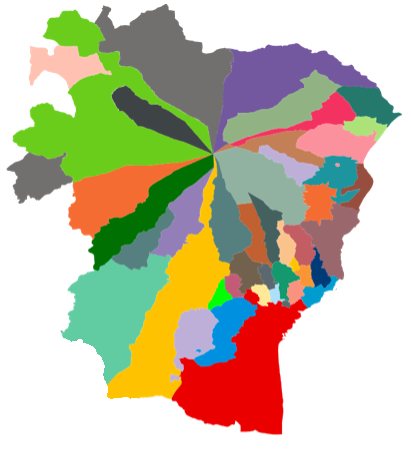}
  \caption{Regions}
  \label{fig:sub1}
\end{subfigure}%
\begin{subfigure}{.33\textwidth}
  \centering
  \includegraphics[width=.9\linewidth]{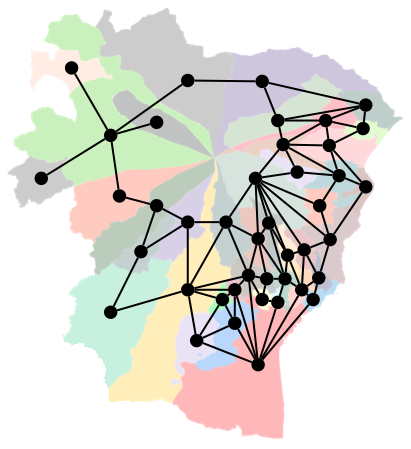}
  \caption{Planar graph}
  \label{fig:sub2}
\end{subfigure}
\begin{subfigure}{.33\textwidth}
  \centering
  \includegraphics[width=.9\linewidth]{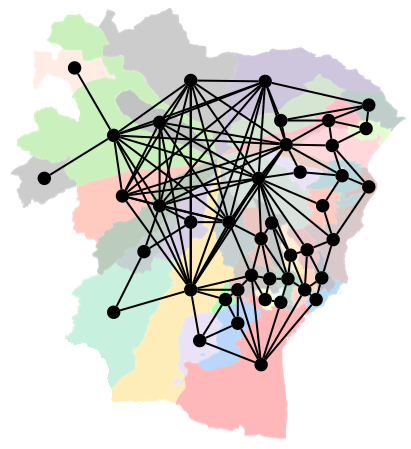}
  \caption{Map graph}
  \label{fig:sub3}
\end{subfigure}
\caption{(a) Some municipalities of the province of Catania (Sicily, Italy). The vertices are representing connected regions. (b) The planar graph has edges between two regions if they share a border. (c) The map graph has edges whenever regions share at least a single point.}
\label{fig:map-graph}
\end{figure}

The problems tractable on planar graphs are often considered also in more general classes of graphs. Most common such classes include bounded genus graphs and even more general H-minor-free graphs. In this work however, we focus on a different generalization, $\ie$ map graphs introduced by Chen et al.~\cite{CGP02}. They are defined as intersection graphs of internally disjoint connected regions in the plane. Unlike for planar graphs, two regions are adjacent if they share at least one point (see Figure~\ref{fig:map-graph}). Notably, map graphs are not H-minor-free as they may contain arbitrarily large cliques as minors.

It is useful to characterize map graphs as \textit{half-squares} of bipartite planar graphs. A half-square of a bipartite graph $W=(V \cup U, E_W)$ is a graph $G=(V,E)$ where we have an edge between a pair of vertices, whenever the distance between these vertices in $W$ is equal to two. If $W$ is planar, then it is called a witness graph of map graph $G$. See Figure~\ref{fig:witness} for a witness graph (solid edges) and the corresponding map graph (dashed edges).

\begin{wrapfigure}{R}{0.3\textwidth}
\centering
\includegraphics[width=0.28\textwidth]{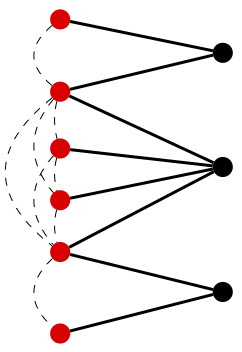}
\caption{Map graph and its witness}
\label{fig:witness}
\end{wrapfigure}

We are the first to study the Steiner tree problem on  map graphs. We study the case when all edges have uniform cost; otherwise the map graphs would capture the general case. To see this, observe that a clique $K_n$ is a map graph and putting large costs on some edges mimics any arbitrary graph.

On the other hand, the case of map graphs with uniform edge costs is still more general than arbitrary edge-weighted planar graphs for the Steiner tree problem. This follows from the fact that subdividing edges preserves planarity, and we can reduce planar graphs to the uniform case.

Therefore it is natural to ask if there is a PTAS for the Steiner tree in our setting. This question gets even more compelling upon realizing, that this is a special case of node-weighted problem on planar graphs. To see this, consider the following reduction: take the witness graph $W=(V \cup U, E_W)$ of the uniformly edge-weighted map graph and put weight $1$ on the vertices in $V$ and weight $0$ on the vertices in $U$. The terminals are kept at the corresponding vertices in $V$. The solutions for the resulting node-weighted problem can be easily translated back to the initial instance. The validity of the reduction is assured by a simple fact; that the number of vertices in a tree is equal to the number of edges in this tree plus one.
The structure of instances arising from this  reduction is very special and is captured in the definition below.

\begin{definition}
The node-weighted Steiner tree instance is \textbf{map-weighted} if it is a bipartite planar graph with weight $1$ on the left side and weight $0$ on the right side. Moreover, the terminals are required to lie on the left side.
\end{definition}

\subsection{Our results}

We study the node-weighted Steiner tree problem on planar graphs and give a PTAS for the special case of map-weighted instances.
\begin{theorem}
\label{thm:main-theorem}
  There is a polynomial-time approximation scheme for node-weighted Steiner tree problem on map-weighted instances.
\end{theorem}

By the reduction described above, we immediately obtain the PTAS for edge-weighted Steiner tree problem on uniform map graphs.

\begin{corollary}
  There is a polynomial-time approximation scheme for the Steiner tree problem on uniform map graphs.
\end{corollary}

In the proof of Theorem~\ref{thm:main-theorem} we adopt the framework for constructing PTASes and the brick-decomposition of Borradaile et al.~\cite{bkm09}.
However, we need to tackle additional obstacles related to high-degree vertices in the node-weighted setting.

The first difficulties emerge in the Spanner construction. In the \textit{cutting-open} step, the duplication of high-degree vertices may make the cost unbounded. Another difficulty is bounding the number of  portals needed. Essentially, the presence of expensive high-degree vertices excludes the existence of nearly-optimum solution with bounded number of joining vertices. The properties of map-weighted instances allow us to overcome multiple difficulties and prove the following.
\begin{lemma}[Steiner-Tree Spanner]
  \label{thm:spanner} Given a map-weighted instance $W = (V \cup U, E_W)$ for a map graph $G$, where $R \subseteq V$ are terminal nodes, there is a polynomial time algorithm which outputs a spanner subgraph $H \subseteq W$ containing all the terminals $R$.
  \begin{enumerate}[(i)]
      \item \label{itm:shortness} (shortness property) $w(H) = f(\eps) \cdot OPT(W, R)$
      \item \label{itm:spanning} (spanning property) $OPT(H, R) \leq (1+\eps) \cdot OPT(W, R)$
  \end{enumerate}
  where $f(\eps)$ is a function which depends only on $\eps$ and $OPT(G, R)$ is the cost of an optimal Steiner tree for graph $G$ and set of terminals $R \subseteq V(G)$.
\end{lemma}

A different trouble comes up in the use of Contraction Decomposition Theorem with node-weights.
A naive approach could be to move the costs of vertices to edges by setting the cost of each edge to be the sum of costs of its endpoints and then using the contraction decomposition theorem as it is.
However --- again due to high-degree nodes --- the total cost of edges would no longer be a constant approximation of OPT. Therefore we cannot directly use the existing contraction decomposition theorem. 

To handle the last issue, we develop a new decomposition theorem with the additional property that each vertex participates in a limited number of sets.

\begin{lemma}[Node-weighted Contraction Decomposition]
\label{thm:contraction-decomposition}
  There is a polynomial time algorithm that given a planar embedding of a graph $G$ and an integer $k$, finds $k$ sets $E_0, E_1, \dots, E_{k-1} \subseteq E(G)$ such that:
  \begin{enumerate}[(i)]
    \item contracting each $E_i$ results in a graph with treewidth $O(k)$, and
    \item for each vertex $v$, all the incident edges of $v$ are in at most two sets $E_{i}, E_{j}$.
  \end{enumerate}
\end{lemma}

We note that our decomposition can be applied to any node-weighted \textit{contraction-closed} problem, $\ie$ the problem for which contracting edges and setting the weight of resulting vertex to $0$ does not increase the value of optimum solution.
Therefore the lemma above adds a novel technique to the existing framework for planar approximation schemes.

Finally, using standard techniques, we give a dynamic programming algorithm for the node-weighted Steiner tree problem on bounded treewidth instances (see Appendix~\ref{sec:DP}).

\begin{lemma}[Bounded Treewidth NWST]
\label{lem:bounded_tw_nwst}
An optimal node-weighted Steiner tree can be found in time $2^{\OO(t\log t)}\cdot n^{\OO(1)}$, where $t$ is the treewidth of the input graph with $n$ vertices.
\end{lemma}

We note that Lemma~2 and Lemma~3 work for arbitrary node-weights. Only the spanner construction of Lemma~1 uses properties of the map-weighted instances.

In the next section we give the details of the Spanner construction. In Section~\ref{sec:contraction-decomposition} we prove Lemma~\ref{thm:contraction-decomposition} and show how the combination of the three above lemmas yields the main result. In the last section we conclude with a puzzling open problem.

\section{Spanner construction for map-weighted graph}
\label{sec:spanner}

In this section we describe how we construct the spanner for a map-weighted planar witness graph $W$ and prove Lemma~\ref{thm:spanner}. For convenience, instead of $(1+\eps)$, we will prove the property~(\ref{itm:spanning}) for $(1+c\eps)$ where $c\geq0$ is some fixed constant. For any given $\eps>0$, running the construction for $\Tilde{\eps} = \eps/c$ gives the precise result. From now on, we will work with a fixed embedding of the witness graph $W$.

\paragraph{Notations:} For any map-weighted graph $W$, we define $d_W \colon V^2 \rightarrow \mathbb{R}$ to be the function giving the node-weighted length of the shortest-path between any two vertices using only the edges from $W$ (including the end vertices weights). Let $P_W(u, v) \subseteq W$ be an arbitrary path of cost $d_W(u,v)$. Similarly, let $\ell_W \colon V^2 \rightarrow \mathbb{R}$ be the length of the unweighted shortest-path ignoring the node-weights between any two vertices using only the edges from $W$. Similarly we define for any path $P \subseteq W$, $c(P)$ to be the cost of the path corresponding to the map-weights (including the end vertices) and $\ell(P)$ to be the length of the unweighted-path ignoring the node-weights. Analogously, for any graph $H \subseteq W$, we define $c(H)$ to be the total weight of nodes of $H$ and $\ell(H)$ to be the number of edges of $H$.

For any path $P$ and $u, v \in V(P)$, we define $P[u, v]$ to be the sub-path starting at $u$ and ending at $v$ (including $u$ and $v$) and $P(u, v)$ to be the sub-path starting at $u$ and ending at $v$ (excluding $u$ and $v$). We refer to any path/cycle with no edges and one vertex as {\em singleton} path/cycle and the ones containing at least one edge as {\em non-singleton} path/cycle respectively.

The spanner construction is summarized in Algorithm~\ref{alg-spanner}.

\begin{algorithm}[ht!]
  \caption{Spanner construction}
  \label{alg-spanner}
  \begin{algorithmic}[1]
    \State Start with a $2.4$-approximate node-weighted Steiner tree solution for graph $W$ and terminal set $R$ using~\cite{BermanY12}.
    \State Cut open the corresponding solution tree $ST$ in $W$ to create another graph $W'$ which has an outerface with boundary $\bset$ of cost at most $10 \cdot OPT$.
    \State Build the Mortar graph $MG$ on the cut-open graph $W'$ using the procedure in Section~6 of~\cite{bkm09}, ignoring the weights on the nodes and using $\ell(e)=1$ for each $e \in E(W')$ and $\eps' := \eps/4$
    \State Construct the set $P(B) \subseteq \partial B$ of portals for each brick $B \in MG$.
    \State For each brick $B \subseteq MG$ and for each subset $X \subseteq P(B)$, run the generalized Dreyfus-Wagner algorithm~\cite{DW71,bwb18} to compute the optimal Steiner tree on terminal set $X$ in map-weighted graph $B$ in time $3^{|X|}n^{\OO(1)}$.
    \State Return the union of $MG$ along with all the trees found in the previous step.
  \end{algorithmic}
\end{algorithm}

\begin{figure}[!b]
\centering
\begin{subfigure}{.40\textwidth}
  \centering
  \includegraphics[width=.6\linewidth]{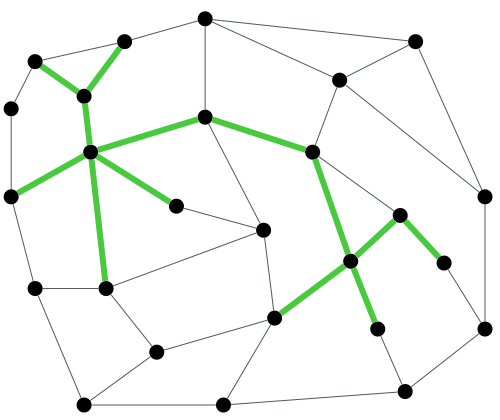}
  \label{fig:cuttingopen1}
\end{subfigure}%
\begin{subfigure}{.40\textwidth}
  \centering
  \includegraphics[width=.6\linewidth]{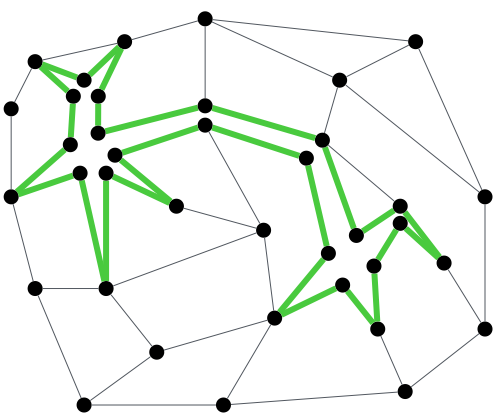}
  \label{fig:cuttingopen2}
\end{subfigure}
\caption{Cutting graph open}
\label{fig:cut-open}
\end{figure}

Before proving Lemma~\ref{thm:spanner}, we elaborate on the steps of Algorithm~\ref{alg-spanner} that require more detailed explanation and state the key properties of the construction.

\paragraph{Cutting-Open operation.} We start with a $2.4$-approximate node-weighted Steiner tree solution $ST$ for our node-weighted plane graph $W$ and terminal set $R$ using~\cite{BermanY12}. Using tree $ST$, we perform an {\em cut-open} operation as in~\cite{bkm09} (see Figure~\ref{fig:cut-open}) to create a new map-weighted planar graph $W'$ whose outer face is a simple cycle $\bset$ arising from $ST$.

Since we are dealing with node weights and the node degrees are unbounded, we need an additional argument to bound the cost of $\bset$ as compared to the edge-weighted case. A crucial property used to prove the observation is that all the leaves of $ST$ have weight one.
\begin{lemma}(Cut-Open)
\label{lem:cut-open}
The cost $c(\bset)$ of the boundary $\bset$ is at most $10 \cdot \opt$. Moreover $R \subseteq V(\bset)$.
\end{lemma}

\paragraph{Mortar Graph Construction} 
We apply the construction of a mortar graph along with a brick decomposition as described in~\cite{bkm09} as a black box. Here, we state the properties of the mortar graph that we need in our work without referring to the details of the algorithm which constructs it.

\begin{enumerate}[(i)]
    \item The mortar graph $MG$ is a subgraph of the cut-open graph $W'$.
    \item Let $f$ be a face of the mortar graph. A brick $B$ (corresponding to $f$) is the subgraph of $W'$ enclosed by the boundary $\partial f$ of $f$. Specifically, the boundary $\partial B$ of $B$ is precisely $\partial f$.
    \item The collection of all bricks covers the cut-open graph $W'$.
    \item The mortar graph is ``grid-like'' in the following sense. The boundary $\partial B$ of each brick $B$ can be decomposed into a western part $W_B$, a southern part $S_B$, an eastern part $E_B$ and a northern part $N_B$ (see Lemma~\ref{lem:brick-prop} below). Each of these parts is close to be a shortest path (see Definition~\ref{def:eps-shortness} below).
\end{enumerate}
The construction of the mortar graph and the corresponding brick decomposition as described in~\cite{bkm09} has two parameters. An error parameter $\eps'$ and an edge-weight function $\ell$. We invoke their construction procedure of the mortar graph as a black box using error parameter $\eps'=\eps/4$ and unit edge-weights $\ell(e)=1$ for all edges $e\in E(W')$. Note that the node weights are ignored in this construction.

In what follows, we will prove certain properties of the mortar graph $MG$ about its node weights and error parameter $\eps$ based on the fact that similar properties hold with respect to the unit edge weights and error parameter $\eps'=\eps/4$.

The following technical lemma tells us that the node weight of a path is roughly half its edge length apart from a small additive offset. It turns out convenient for the shortness properties of the spanner that this offset is the same for any two paths sharing their end nodes.

\begin{lemma}\label{lem:length-vs-cost}
Let $P,P'$ be two paths sharing both of their end points $u$ and $v$. Then the following properties hold.
\begin{enumerate}[(i)]
    \item There is $b\in\{0,-1,-2\}$ such that $\ell(P)=2c(P)+b$ and $\ell(P')=2c(P')+b$.
    \item $\ell(P)/2\leq c(P)\leq \ell(P)/2+1$
    \item $\ell_W(u,v)/2\leq d_W(u,v)\leq \ell_W(u,v)/2+1$
    \item $P$ is a shortest path under $\ell_W$ if and only if it is a shortest path under~$d_W$.
\end{enumerate}
\end{lemma}

The following lemma gives cost bounds on the mortar graph. In contrast to~\cite{bkm09}, we have to exclude singleton boundaries in property~(i) in order to avoid a cost explosion. To account for the singleton boundaries in the shortness property of~Lemma~\ref{thm:spanner} we bound their total number separately. (See proof of~Lemma~\ref{thm:spanner}.)
\begin{lemma}\label{lem:cost-mortargraph}
  The mortar graph $MG$ has the following two properties.
  \begin{enumerate}[(i)]
      \item The total cost $\sum_{B\in\mathcal{B}: E(W_B) \neq \varnothing}c(W_B) + \sum_{B\in\mathcal{B}: E(E_B) \neq \varnothing}c(E_B)$ of all the non-singleton western and eastern boundaries of all bricks is bounded by $O(\eps)\cdot\opt$.
      \item The total cost $c(MG)$ of the mortar graph is $O(1/\eps)\cdot \opt$.
  \end{enumerate}
\end{lemma}

\paragraph{Designating Portals.}
For finding the portals, we use directly the same greedy procedure as in Step 3(a)~\cite{bkm09}, as it does not work for bricks having boundary with small cost, because of the additive one in Lemma~\ref{lem:length-vs-cost} bound. To circumvent this issue, we pick {\em all} the vertices to be the set of portals when the boundary is cheap. And then for any remaining brick, the boundary cost is bounded from below. For these bricks the greedy procedure works, as the additive plus can be absorbed in the big-Oh by creating a factor 3 gap in the number of portals and the cost bounds, which is sufficient.

By balancing all the parameters, we get that for any brick $B$ in Mortar graph $MG$ there exists at most $3\theta$ portals $P(B)$ such that each vertex on the boundary of $B$ lies within a distance of at most $c(\partial B)/\theta$ from some portal. Here $\theta = \theta(\eps) = \Theta(g(\eps) \eps^{-2})$, where $g(\eps)$ is defined in Lemma~\ref{thm:brick-structural}.
\begin{lemma}
\label{lem:portals} Given a brick $B$ of the Mortar graph $MG$, there exists a set of vertices $P(B) \subseteq \partial B$, such that:
\begin{enumerate}
    \item (\textbf{Cardinality Property}) $|P(B)| \leq 3\theta$
    \item (\textbf{Coverage Property}) For any $u \in \partial B$, there exists $v \in P(B)$, such that $d_{\partial B}(u, v) \leq  c(\partial B)/\theta$ and $\ell_{\partial B}(u, v) \leq \ell(\partial B)/(3\theta)$
\end{enumerate}
\end{lemma}

Now we can sketch the proof of Lemma~\ref{thm:spanner}.
\begin{proof}[Proof of Lemma~\ref{thm:spanner}]

  (\ref{itm:shortness}) \textbf{Shortness property.}
  We have to bound the total cost of $H$ which consists of the mortar graph and optimal Steiner trees added in step 5 of Algorithm~\ref{alg-spanner}. By Lemma~\ref{lem:cost-mortargraph} the cost of the mortar graph is $O(1/\eps)\cdot \opt$. We bound the cost of Steiner trees analogously as in the Lemma 4.1~\cite{bkm09}, \ie we charge it to the cost of the mortar graph (losing a large constant). However, we have to take extra care to not overcharge vertices adjacent to multiple bricks.

  Consider any brick $B$ and any tree connecting portals of $B$ added in step 5. The cost of this tree can be upper-bounded by the cost of the boundary of the brick $c(\partial B)$. Since there is a constant number of such trees (this follows from Lemma~\ref{lem:portals}), the total cost of the trees added is constant times the cost of the boundary of the brick. Now, if every vertex belonged to the boundary of a constant number of bricks, this would imply that the total cost of all Steiner trees is bounded by constant times $\opt$. Below we show that if it is not the case for some vertices, then we have a different way to pay for the cost incurred by these vertices.

  We say that a vertex $v$ is a {\em corner} of a brick if it belongs to the intersection of $N$ (or $S$) with $E$ (or $W$). In a special case in which $E$ or $W$ is empty, we call $v$ which belongs to the intersection of $N \cap S$ a {\em trivial corner}. We also say, that $v$ is a {\em regular} boundary vertex of a brick if it is not a corner of this brick.
  
  Observe that $v$ can be a regular boundary vertex of at most two bricks. It remains to show how to charge corner vertices. For trivial corners, observe that there is as many unique pairs (corner vertex, corresponding brick) as there were strips during creation of the mortar graph. Note that there are $O(f(\eps)OPT)$ strips (see Lemma~\ref{clm:strips-count} in Appendix~\ref{sec:ommited-proofs}). This, together with the fact that weight of each vertex is at most $1$ implies that we charge at most constant times $\opt$ for trivial corners.

  The charging for non-trivial corners is different. By the first property of Lemma~\ref{lem:cost-mortargraph}, we know that the sum of the costs of west and east boundaries for all bricks is bounded by $O(\eps) \opt$. As non-trivial corners belong to $W$ or $E$, the total cost incurred by charging to non-trivial corners is also bounded by constant times $\opt$. This finishes the proof of the shortness property.

  (\ref{itm:spanning}) \textbf{Spanning property}
  The proof of the spanning property is in the spirit similar to the proof of Structural Theorem 3.2, and Lemma 4.2 in~\cite{bkm09}. However, we cannot use their approach of portal-connected graph via the {\em brick insertion} operation, as in the node weighted setting this would destroy the structure of the optimum solution. Therefore we give a slightly more direct proof, where we avoid the portal-connected graph at all.

  Moreover, we have to take extra care when showing structural lemmas. These proofs do not transfer immediately to the node-weighted instances. For example, we have to heavily use special structure of map-weighted instances to bound the number of {\em joining vertices}. Due to the high technicality of the arguments and a lack of space, the details are explained in Appendix~\ref{sec:proof-spanning}.
\qed
\end{proof}

\section{Node-weighted Contraction Decomposition}
\label{sec:contraction-decomposition}
    In this section we give a reduction of a spanner to graphs with bounded treewidth. The input to our reduction is a spanner ($\eg$ the one constructed in the previous section, see Lemma~\ref{thm:spanner}), $\ie$ a graph $H$ of cost $f(\epsilon)\cdot \opt$ that approximately preserves an optimum solution.

    We apply Lemma~\ref{thm:contraction-decomposition} (proven later in this section) with $k=\frac{2 \cdot f(\epsilon)}{\epsilon}$ to graph $H$ and obtain sets $E_0, E_1, \ldots, E_{k-1}$.
    Now, define the cost of the set of edges to be the total weight of vertices incident to edges in this set.
    Because every vertex belongs to at most two sets, the total cost of all the edge sets is at most
    $2 f(\epsilon) \cdot \opt$ and therefore, the cheapest set, say $E_c$ has cost at most $\epsilon \cdot \opt$.

    We now contract $E_c$ to obtain graph $H'$. We assign weight $0$ to the vertices resulting from contraction, while the weight of the other stays untouched. It is clear that after this operation the value of the optimum solution will not increase. Now we solve the node-weighted Steiner tree problem for $H'$ using Lemma~\ref{lem:bounded_tw_nwst} (see Appendix~\ref{sec:DP}). We can do this in polynomial time, since the treewidth of $H'$ is at most $k$. We include the set of edges $E_c$ in the obtained solution for $H'$ to get the final tree of cost $(1+\epsilon) \cdot \opt$.

    Therefore we are left with proving Lemma~\ref{thm:contraction-decomposition}. 
    Here, we build on Klein's~\cite{Klein08} contraction decomposition and modify it to our needs.

    \begin{proof}[Proof of Lemma~\ref{thm:contraction-decomposition}]

    \begin{figure}[t!]
    \centering
    \begin{subfigure}{.3\textwidth}
      \centering
      \includegraphics[width=.7\linewidth]{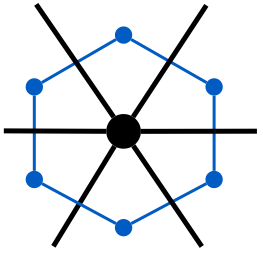}
      \caption{}
      \label{fig:nwcd1}
    \end{subfigure}%
    \begin{subfigure}{.3\textwidth}
      \centering
      \includegraphics[width=.53\linewidth]{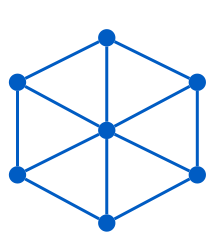}
      \caption{}
      \label{fig:nwcd2}
    \end{subfigure}
    \caption{(a) A vertex of a graph (in black) and the dual. The dual is shown in  blue. (b) Triangulation of the dual graph.}
    \label{fig:triangulating}
    \end{figure}

    At first, we triangulate the dual graph $H^*$ by adding an artificial vertex in the middle of each face of the dual graph and introducing artificial edges (see Figure~\ref{fig:triangulating}). This is the crucial step which --- as explained later --- enables us to control the level of edges in the breadth-first search tree.

    Let $J^*$ be the graph after above modification of $H^*$. Now, from some arbitrary node $r$ we run breadth-first search on $J^*$, this gives us a partition of $V(J^*) = L_0 \cup L_1 \ldots \cup L_d$, for some $d \in \mathbb{Z}^+$. Fix an  $i \in \{0,\ldots, k-1\}$, and let $E^*_i$ be the set of edges with one endpoint in $L_p$ and  the other endpoint in $L_{p+1}$, for all $p$ congruent to $i$ modulo $k$. Also, let $E_i$ be the set of primal edges of $H$ corresponding to the dual edges in $E^*_i$. Note that we do not include in $E_i$ the artificial edges. We claim that sets $E_i$ satisfy both the requirements.

    First we show that each vertex participates in at most two sets.
    \begin{lemma}
    \label{lem:decomposition-at-most-two-sets}
      For each vertex $v \in H$, all the incident edges of $v$ are contained in at most two sets $E_{i}, E_j$.
    \end{lemma}
\begin{proof}
Consider the faces $f_1, f_2, \dots f_l$ incident to a vertex $v$ (in the clockwise order of appearing in the planar embedding). Each $f_i$ has a corresponding vertex in the dual graph $H^*$ and therefore also in $J^*$. Moreover, there is a cycle on these vertices. Call the edges of this cycle $e^*_1, e^*_2, \dots e^*_l$. These edges correspond to all primal edges $e_1, e_2, \dots e_l$ incident to vertex $v$ in $H$. However, $J^*$ has also additional vertex $g$ adjacent to all vertices $f_i$. Therefore, the distance between any two $f_i$ and $f_j$ in $J^*$ is at most 2. Hence, the vertices $f_i$ will be in at most three consecutive layers of BFS ordering. Therefore, all the incident edges of $v$ will be contained in at most two of the $E^*_i$'s, and hence be contained in at most two corresponding $E_i$'s. 
\qed
\end{proof}
    We are left with showing that contracting each set reduces the treewidth. In essence, we use the argument of Klein. We only need to take care of artificial edges.
    \begin{lemma}
    \label{lem:decomposition-has-bounded-tw}
      The graph $H$ after contracting $E_i$ has treewidth $O(k)$.
    \end{lemma}
    \begin{proof}
      Let $J$ be a dual graph of $J^*$. We will call $J$ the primal of $J^*$.

      By directly applying the result of Klein, contracting all the primal edges of $E^*_i$ from $J$ results in a graph $X$ of treewidth $O(k)$. It is easy to see, that contracting all other artificial edges in $X$ results exactly in a graph $H_{/E_i}$. Since contraction of edges does not increase treewidth, the lemma follows.
    \qed
    \end{proof}
    The algorithm described above together with Lemma~\ref{lem:decomposition-at-most-two-sets} and Lemma~\ref{lem:decomposition-has-bounded-tw} completes the proof of Lemma~\ref{thm:contraction-decomposition}.
    \qed
    \end{proof}

\section{Conclusion}
\begin{wrapfigure}{t}{0.3\textwidth}
\centering
\vspace{-23pt}
\includegraphics[width=0.3\textwidth]{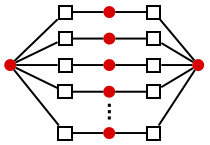}
\caption{Node-weighted subset spanner does not exists. Red vertices are of cost 1, squares are terminals of cost 0.}
\vspace{-10pt}
\label{fig:nosubsetspanner}
\end{wrapfigure}
We reduced the node-weighted Steiner tree problem on planar graphs to the problem of finding a spanner.
The main obstacles in constructing such general spanner are caused by the high-degree vertices. The first difficulty arises already in the \textit{cutting-open} step. The second issue is related to bounding the number of joining vertices. 

As we have shown, both of the difficulties are solvable in map-weighted graphs. However, we pose now an interesting open problem: decide the existence of a PTAS for node-weighted Steiner tree problem on map-weighted graphs where terminals are allowed to lie also on vertices of weight $0$.

There are two reasons why the above open problem is compelling. First, it nicely isolates the first difficulty - the latter issue is not present on such graphs. Second, the node-weighted \textit{subset spanner} does not exist for these instances. The subset spanner is a cheap subgraph (in terms of optimum Steiner tree) that approximately preserves distances between pairs of terminals. For contrast, subset spanner construction exists in the edge-weighted case \cite{Klein06}.

Figure~\ref{fig:nosubsetspanner} gives the example of such node-weighted instance. The cheapest Steiner tree has cost equal to 2. However, the subset spanner would have to use all the red central vertices.

\bibliographystyle{splncs04}
\bibliography{latin}

\appendix

\section{Proof of Lemma~\ref{lem:bounded_tw_nwst}}
\label{sec:DP}


Here we show how to obtain a polynomial time algorithm for computing the minimum node weighted Steiner network for a graph $G$ with constant treewidth $t$, whose tree decomposition ${\cal T}=(T,\{X_b\}_{b\in V(T)})$ is given to us as input. We note that an approximate tree decomposition of a planar graph can be obtained in linear time~\cite{kammer2016approximate}. We are given a set of terminals $R$ which need to be connected together and the weight function on the nodes is $w:V(G) \rightarrow \mathbb{R}^+$. We would be doing dynamic programming over the bags of the tree decomposition. The algorithm we present here is analogous to the algorithm for the edge weighted version~\cite{chimani2012improved}. As we were unable to find any algorithm for the node weighted variant in the literature, we present the algorithm for completeness. We start with the definition of tree decomposition.
\begin{definition}[Tree Decomposition]
For a graph $G$, a pair ${\cal T} = (T,\{X_b\}_{b\in V(T)})$ is called as the tree decomposition of $G$, if $T$ is a tree, and the following conditions are satisfied.
\begin{enumerate}
    \item $\bigcup_{b \in V(T)} X_b = V(G)$,
    \item if $(u,v)\in E(G)$, then there is a node $b\in V(T)$, such that $u,v \in X_b$, and
    \item if $u \in X_a$ and $u \in X_b$, then $u$ is contained in $B_z$, for  each $z$ on the unique path from $a$ to $b$ in $T$. \end{enumerate}
    The treewidth of ${\cal T}$ is defined to be ${\rm tw}({\cal T}) = \max_{b \in V(T)} |X_b| -1$.
\end{definition}

We would be using the following  special form of tree decomposition for which dynamic programming algorithms are much simpler and cleaner.

\begin{definition}[Nice Tree Decomposition] A tree decomposition 
${\cal T} = (T,\{X_b\}_{b\in V(T)})$ 
of a graph $G$ is called as a nice tree decomposition if $T$ is a tree rooted at $r$, with $X_r=\emptyset$, and each bag of ${\cal T}$ is one of the following types.
\begin{description}
    \item [Leaf Node] A leaf $\ell \in V(T)$, with $X_\ell = \emptyset$.
    \item [Introduce Node] A non-leaf node $a \in V(T)$, with exactly one child $b$ for which $X_a = X_b \cup \{u\}$, for some $u \in V(G)$, with $u \notin X_b$.
    \item [Forget Node] A non-leaf node $a \in V(T)$ with exactly one child $b$,  for which $X_a = X_b \setminus \{u\}$, for some $u \in V(G)$, with $u \in X_b$.
    \item[Join Node] A non-leaf node $a \in V(T)$ with two children $b,c$, such that $X_a=X_b=X_c$.
\end{description}
\end{definition}

We remark that any tree decomposition can be converted into a nice tree decomposition in polynomial time without increasing the tree width. 
To each node $b\in V(T)$, we associate $G_b$, the subgraph of $G$ induced over the union of the bags corresponding to the nodes contained in the subtree $T_b$. 
We make some simplifying assumptions about the tree decomposition of an instance on which we want to compute the minimum cost node weighted Steiner tree, $(i)$ the terminals are all degree $1$, this can be achieved by placing a node $s'$ in place of an existing terminal $s$ with cost equal to $0$, and putting the terminal $s$ as a leaf node adjacent to $s'$, $(ii)$ the tree decomposition has a terminal present only in a leaf node or the root node, and $(iii)$ the root $r\in K$, is a terminal.

The dynamic programming table $C$ for each node $b\in V(T)$, is indexed by a tuple $(b,S,P)$, where $S \subseteq X_b$, and $P = \{S_1,\dots, S_q\}$ is a partition of $S$. A table entry 
$C(b,S,P)=$ cost of minimum cost node weighted tree $N$ connecting terminals in $R \cap V(G_b)$, satisfying the following properties.
\begin{itemize}
    \item $X_b \cap V(N) = S$, and
    \item $N$ has exactly $q$ components, which can be ordered as $C_1, \dots, C_q$ such that $S_i = V(C_i) \cap X_b$, for $i \in [q]$.
\end{itemize}

As we are interested in a minimum cost tree, we would store a value of $+\infty$ in $C(b,S,P)$ to mark the fact that a tree satisfying the definition  above does not exist.
The weight of optimum Steiner tree will be stored in $C(r,\{r\},\{\{r\}\})$, where $r$ is the root of $\cal{T}$. Next, we show how to fill the dynamic programming table $C$.

\begin{description}
    \item[Leaf Node]
    Suppose $t\in T$ is a leaf node with $X_t=\{u\}$. Two cases arise, if $u \in R$, then $C(t,\emptyset,\{\emptyset\})=+\infty$ and $C(u,\{u\},\{\{u\}\})=w(u)$. Otherwise, $C(t,\emptyset,\emptyset)=0$ and $C(u,\{u\},\{\{u\}\})=w(u)$. 
    
    \item [Introduce Node] Let $a$ be an introduce node with exactly one child $b$ for which $X_a = X_b \cup \{u\}$, for some $u \in V(G)$, with $u \notin X_b$. Recall that due to the properties satisfied by the tree-decomposition $\cal{T}$, $u$ can not be a terminal, and we are not forced to include it in a solution, like in the leaf node above. As the edges are free of cost we can always include all edges incident to solution vertices. Suppose we want to fill an entry $C(a,S,P)$, with $P = \{S_1,\dots, S_q\}$. If $u\notin S$, then $C(a,S,P) = C(b,S,P)$. If $u \in S$, then let $u\in S_i$, for some $i\in [q]$.
    If $N_{G_b}[u]\cap X_b \not \subseteq S_i$, then $C(a,S,P) = +\infty$. Otherwise, fix $C(a,S,P) = w(u) + C(b,S\setminus \{u\},(P\setminus \{S_i\}) \cup \{S_i\setminus \{u\}\})$. 
    
    \item [Forget Node] A non-leaf node $a \in V(T)$ with exactly one child $b$, for which $X_a = X_b \setminus \{u\}$, for some $u \in V(G)$, with $u \in X_b$. If a solution $N$ does not use the vertex $u$, then the cost is $C(b, S, P)$. If a partial solution uses the vertex $u$, then we need to look up all the table entries which contained $u$. This gives us the formula
    \[
    C(a, S, P) = \min \{ \min_{
    P'\mbox{\footnotesize :a partition of }(S\cup \{u\})
    } 
    C[b,S\cup\{u\}, P')],
    C[b,S,P]
    \}.
    \]
    \item[Join Node] For a non-leaf node $a \in V(T)$ with two children $b,c$, such that $X_a=X_b=X_c$, it is enough to make sure that the partitions of the merged bags are same, along with subtracting the double counted vertices. This gives us the formula
    \[
    C(a, S, P) = C(b, S, P) + C(c, S, P) - w(S).
    \]
\end{description}

The proof of the algorithm presented above is analogous to the edge weighted variant, we refer the reader to~\cite{chimani2012improved} for details. The running time of the algorithm is as claimed as the number of partitions of the vertices in a bag of the tree decomposition is $2^{\OO (t \log t)}$.
This completes the proof of Lemma~\ref{lem:bounded_tw_nwst}.


\section{Ommited proofs}
\label{sec:ommited-proofs}

\noindent \textbf{Lemma~\ref{lem:cut-open}.}
\textit{
The cost $c(\bset)$ of the boundary $\bset$ is at most $10 \cdot \opt$. Moreover $R \subseteq V(\bset)$.
}
\begin{proof}[Proof of Lemma~\ref{lem:cut-open}]
The second property is true by construction. 
The cost of $\bset$ is $\sum_{v \in V(ST)} \alpha(v)w(v)$, where $\alpha(v)$ is the number of copies of any vertex $v \in ST$ in $\bset$. Note that $\alpha(v) = d_{ST}(v)$, where $d_{ST}(v)$ is the degree of $v$ in the Steiner tree $ST$. Let us fix any node $r \in V(ST)$ such that $d_{ST}(r) > 1$ to be the root node and let $U \subseteq V(ST)$ be the set of leaves of $ST$. WLOG, we can assume that all leaves in $U$ belong to $T$. This implies that for each $v \in U$, $w(v) = 1$. An easy counting argument implies that $|U| = d_{ST}(r) + \sum_{v \in ST - (U \cup \{r\})} (d_{ST}(v) - 2)$. Let $V_2 \subseteq V(ST) - \{r\}$ be all the vertices of degree two and let $I = ST - (U \cup \{r\} \cup V_2)$. Clearly, for any $v \in I$, the degree in $ST$ is at least three. This fact and the above equality implies that $|U| \geq 1 + \sum_{v \in I} 1 = 1 + |I|$. Putting all these things together, we get
\begin{align*}
    c(\bset) &= \sum_{v \in V(ST)} \alpha(v)w(v) = \sum_{v \in V(ST)} d_{ST}(v) w(v) \\
             &= \sum_{v \in V_2} d_{ST}(v) w(v) + \sum_{v \in V(ST) - V_2} d_{ST}(v) w(v) \\
             &\leq 2 w(V_2) + \sum_{v \in V(ST) - V_2} d_{ST}(v) \\
             &\leq 2 w(V_2) + |U|+ d(r) + \sum_{v \in I} d_{ST}(v) \\
             &\leq 2 w(V_2) + |U|+ d(r) + \sum_{v \in I} (d_{ST}(v) - 2) + 2 |I| \\
             &\leq 2 w(V_2) + |U| + |U| + 2 |U| \\
             &\leq 4 c(ST) \leq 10 \cdot \opt
\end{align*}
The second last inequality follows because $w(U) = |U|$, $V_2 \cap U = \emptyset$. The last inequality follows because $ST$ was a $2.4$-approximate solution.
\qed
\end{proof}

\noindent \textbf{Lemma~\ref{lem:length-vs-cost}.}
\textit{
Let $P,P'$ be two paths sharing both of their end points. Then the following properties hold.
\begin{enumerate}[(i)]
    \item There is $b\in\{0,-1,-2\}$ such that $\ell(P)=2c(P)+b$ and $\ell(P')=2c(P')+b$.
    \item $\ell(P)/2\leq c(P)\leq \ell(P)/2+1$
    \item $\ell_W(u,v)/2\leq d_W(u,v)\leq \ell_W(u,v)/2+1$
    \item $P$ is a shortest path under $\ell_W$ if and only if it is a shortest path under~$d_W$.
\end{enumerate}
}
\begin{proof}
   To see the first property note that $P$ and $P'$ both correspond to sequences of node weights alternating between $0$ and $1$. The property follows since both sequences have the same start and end value since $P$ and $P'$ share their end nodes.
   
   The second property is a direct consequence of the first one by plugging $b=0$ and $b=-2$ to $\ell(P)/2-b/2=c(P)$. The third property is a direct consequence of the second one. The third property also follows from the first property by assuming that $P$ is a shortest path w.r.t.\ to $\ell_W$ or $d_W$.
 \qed
 \end{proof}

For map-weighted graph $W'$, the following claim is easy to see.

\begin{observation}
\label{clm:strict-paths}
If $P,P'$ are any two paths between any two vertices $u, v $ in $W'$ such that $\ell(P) < \ell(P)'$, then $V(P) \leq V(P') - 2$.
\end{observation}

Let $H$ be the graph we get after adding the shortcuts to the boundary $\bset$ which decomposes $W'$ into the strips. 

\begin{lemma}
\label{clm:strips-count}
The total number of strips is bounded by $|V(\bset)|/2+1$. 
\end{lemma}
\begin{proof}
 We can view the strip decomposition process described in Step~1(c),~\cite{bkm09} as reducing the length of the {\em unprocessed} part of the inside of boundary $\bset$.
 Initially the length of the boundary $\ell(\bset)$ is $|V(\bset)|$. Let $H = \bset$ be the initial graph to which we will add the shortest paths to create strip decomposition. As soon as the process adds a shortest path $N$ from $W'$ between vertices $x, y \in V(H)$ to create a strip, then the length of the part enclosed by $N$ and $H - H [x, y]$ shrinks by at least two. 
 This is true by Observation~\ref{clm:strict-paths}, since the shortest path $N$ is added iff the length of $H [x, y]$ is strictly longer than $N$.
 Now since the length can decrease by at most $|V(\bset)|$ times, hence the process can add at most $|V(\bset)|/2$ shortcuts. The bound on strips has a plus one to account for the last strip created simultaneously with the second last one.
 
 \end{proof}

\begin{lemma}
\label{lem:strips-cost} The total cost of strips is bounded by $O(\eps^{-1}) OPT$.
\end{lemma}
\begin{proof}
    Using Lemma~6.3,~\cite{bkm09} which says that the total length of strips, i.e. $\ell(H) = |E(H)|$ is at most $(1+\eps^{-1}) \ell(\bset)$, implies that $|V(H)|$ is at most $(1+\eps^{-1}) \ell(\bset)$, since $H$ is a connected graph. This together with the fact that $c(H) \leq |V(H)|$ and $\ell(\bset) \leq 2 c(\bset) = O(OPT)$, implies the lemma.
\end{proof}

\noindent \textbf{Lemma~\ref{lem:cost-mortargraph}.}
\textit{
  The mortar graph $MG$ has the following two properties.
  \begin{enumerate}[(i)]
      \item The total cost $\sum_{B\in\mathcal{B}: E(W_B) \neq \varnothing}c(W_B) + \sum_{B\in\mathcal{B}: E(E_B) \neq \varnothing}c(E_B)$ of all the non-singleton western and eastern boundaries of all bricks is bounded by $O(\eps)\cdot\opt$.
      \item The total cost $c(MG)$ of the mortar graph is $O(\eps^{-1})\cdot \opt$.
  \end{enumerate}
}
\begin{proof}

   The two properties hold with respect to the edge weights by Lemma~6.6 and Lemma~6.9 in~\cite{bkm09}. The proofs rely on the fact that the length of the boundary $\bset$ is $O(OPT)$. The same is true for node weights by Lemma~\ref{lem:cut-open}. Hence it is enough to bound the costs by the boundary cost. They show that $\sum_{B}\ell(W_B\cup E_B) = O(\eps) \ell(\bset)$ and $\ell(MG) = O(1/\eps) \ell(\bset)$.
   Since the boundary $\bset$ has at least one edge, hence $\ell(\bset) \leq 2 c(\bset)$. For any non-singleton path $P$, $\ell (P) \geq 1$, hence using Lemma~\ref{lem:length-vs-cost}, this implies that $c(P) \leq \ell (P)/2 + 1 \leq 3\ell(P)/2$.
   This observation implies that $\sum_{B\in\mathcal{B}: E(W_B) \neq \varnothing}c(W_B) + \sum_{B\in\mathcal{B}: E(E_B) \neq \varnothing}c(E_B) \leq \frac{3}{2} \sum_{B\in\mathcal{B}: E(W_B) \neq \varnothing}\ell(W_B) + \sum_{B\in\mathcal{B}: E(E_B) \neq \varnothing}\ell(E_B) = O(\eps) \ell(\bset)$.
    
    For the second part, we note that all the north, south, singleton west and singleton east boundaries are included in the strips. Hence we get the bound on the total cost $c(MG)$ by combining the bounds for strips from Lemma~\ref{lem:strips-cost} and the bound on the non-singleton west and east boundaries.
\qed
\end{proof}

\noindent \textbf{Lemma~\ref{lem:portals}.}
\textit{
Given a brick $B$ of the Mortar graph $MG$, there exists a set of vertices $P(B) \subseteq \partial B$, such that:
\begin{enumerate}
    \item (\textbf{Cardinality Property}) $|P(B)| \leq 3\theta$
    \item (\textbf{Coverage Property}) For any $u \in \partial B$, there exists $v \in P(B)$, such that $d_{\partial B}(u, v) \leq  c(\partial B)/\theta$ and $\ell_{\partial B}(u, v) \leq \ell(\partial B)/(3\theta)$
\end{enumerate}
}
\begin{proof}
  We assume $|V(\partial B)| > 3\theta$, otherwise let $P(B) = V(\partial B)$ and the lemma follows trivially. This implies that $2 c(\partial B) = |V(\partial B)| > 3\theta$.
  
  
  Using the same greedy procedure as in Step 3(a),~\cite{bkm09} on the unit edge weighted graph $B$ by ignoring the node weights and plugging in $3\theta$ we get our set of portals $P(B)$. We get the cardinality bound on $P(B) \leq 3\theta$ and $\ell_{\partial B}(u, v) \leq \ell(\partial B)/(3\theta)$ by Lemma 7.1 and 7.2,~\cite{bkm09}.
  
  Now for any vertex $u \in \partial B$, there exists $v \in P(B)$, such that $\ell_{\partial B}(u, v) \leq \ell(\partial B)/(3\theta)$.
  We get that $d_{\partial B}(u, v) \leq \ell_{\partial B}(u, v)/2 +1 \leq \ell(\partial B)/(6\theta) + 1 \leq c(\partial B)/(3\theta) + 1 \leq c(\partial B)/\theta$. The first and the second last inequalities follow from Lemma~\ref{lem:length-vs-cost}. The last inequality follows from the fact that $c(\partial B) > \frac{3}{2}\theta$.
\qed
\end{proof}

\section{Proof of the spanning property}
\label{sec:proof-spanning}

First we have to prepare the ground by stating properties and lemmas. These will be used in the proof of the spanning property which can be found at the end of this section.


\begin{definition}[$\eps$-Shortness]\label{def:eps-shortness}
A path $P$ in graph $W$ has $\eps$-shortness property in $W$ or it is $\eps$-short, if for any two vertices $u, v \in V(P)$, $d_P(u, v) \leq (1+\eps)d_W(u, v)$.
\end{definition}

According to Lemma~\ref{lem:length-vs-cost} the node weight of a path is roughly half its edge length apart from a small additive offset. Since this offset is the same for paths sharing end points, the offset does not affect the shortness property of the brick boundary (first claim of the following lemma). However, this does not hold when we consider the second claim as we compare here the node weights of a path $P$ along the south boundary with a path $P'$ connecting south with north boundary. The idea is that when setting $\eps'=\eps/4$ the path $P'$ crossing the brick must have length at least $4/\eps$ so that the impact of the offset is low. 

\begin{lemma}\label{lem:brick-prop} The boundary of any brick $B$ can be partitioned into sub-paths $W_B \cup S_B \cup E_B \cup N_B$, such that:
\begin{enumerate}
    \item $N_B$ has $0$-shortness, $S_B$ has $\eps$-shortness property in $B$.
    \item There exists a number $k \leq \kappa=\Theta(1/\epsilon^3)$ and vertices $s_0, s_1, \dots s_k$ ordered from west to east along $S_B$, such that, for any vertex $x$ of $S_B(s_i, s_{i+1})$: $d_{S_B}(x, s_i)\leq \eps d_B (x, N_B)$.
\end{enumerate}
\end{lemma}
\begin{proof}
   From Lemma 6.10 in~\cite{bkm09}, we know that both claims of the lemma hold w.r.t. $\ell_B$ and error parameter $\eps'=\eps/4$. We show that this implies the claimed properties also for $d_W$ and error parameter $\eps$.

   $N_B$ is a shortest path w.r.t.\ $\ell_W$. By Lemma~\ref{lem:length-vs-cost} it is also a shortest path w.r.t. $d_W$.

   Now we prove that $S_B$ has the $\eps$-shortness property. Let $u,v\in S_B$ be two distinct nodes. And let $P$ be a shortest $u$--$v$ path in $B$ w.r.t.\ to the node weights. We show that $S_B$ fulfills $\eps$-shortness w.r.t.\ $d_B$. Since $P$ and $S_B[u,v]$ share their end nodes by Lemma~\ref{lem:brick-prop} there is $b\in\{0,-1,-2\}$ such that $\ell(P)=2c(P)+b$ and $\ell(S_B[u,v])=2c(S_B[u,v])+b$
   \begin{align*}
       c(S_B[u,v]) & = (\ell(S_B[u,v])-b)/2\\
                   & \leq ((1+\eps/4)\ell(P)-b)/2\\
                   & = c(P)+\eps/8\ell(P)\\
                   & = c(P)+\eps/4(\ell(P)/2-b/2)+b\eps/8\\
                   & = (1+\eps/4)c(P)+b\eps/8\\
                   &\leq (1+\eps/4)d_B(u,v)\, .
   \end{align*}

  For the second claim, verify that
   \begin{align*}
       d_{S_B}(x,s_i) & \leq \ell_{S_B}(x,s_i)/2+1\\
                   & \leq (\eps/4)\ell_B(x,N_B)/2+1\\
                   & \leq (\eps/4) d_B(x,N_B)+1\\
                   & \leq \eps\cdot d_B(x,N_B)\, .
   \end{align*}
   
   To see the last inequality, recall that $x\neq s_i$ and hence $1\leq\ell_{S_B}(x,s_i)\leq (\eps/4)\cdot \ell_B(x,N_B)=(\eps/2)\cdot \ell_B(x,N_B)/2\leq (\eps/2)\cdot d_{B}(x,N_B)$.
\qed
\end{proof}

We use a series of lemmas to simplify any tree $T$ which lies inside a brick of the mortar graph, the simplified tree is slightly costlier by an additive factor of $\eps$. We prove several structural lemmas depending on the need to preserve some vertices for connectivity across to the other boundaries.

We start with the property called as Span 0, here we assume that we have a tree $T$ rooted at $r$ which is part of the embedded map-weighted graph $W$ and all the leaves of $T$ lie on an $\eps$-short path $P$. Also, we do not need to keep $r$ in the simplified tree. The proof is similar to edge weighted case as $\eps$-shortness suffices for a proof.

\begin{lemma}[Span $0$]
\label{lem:span0}
There is a sub-path $P' \subseteq P$ such that $c_W(P') \leq (1+\eps) c_W(T)$.
\end{lemma}
\begin{proof}
  The proof follows directly from the $\eps$-shortness of $P$ by taking $P'$ to be the shortest subpath of $P'$ that spans all vertices in $T \cap P$.
\qed
\end{proof}

For any subgraph $H$ of $W$ and a path $P$ in $W$, a {\em joining vertex} is any vertex in $P$ that has an edge incident from some vertex in $H - P$.

We also have the following observation for any path in $W$. The map weights have a nice property of paths having the zero/one weights occurring alternately on them.

\begin{observation}
\label{obs:alternation_bip}
Let $W=(U\cup V, E_W)$  be any bipartite graph with weights $w(u)=1$, for each $u\in U$, and $w(v)=0$, for each $v\in V$. Then, any path in $W$ consists of vertices with weights alternating between zero and one. Moreover, if we have $d_0$ vertices with weight $0$ on any path $P$, then $c_W(P) \geq d_0-1$.
\end{observation}

Next we prove Span 1 which is analogous to Lemma 10.4 of~\cite{bkm09}, here we need to preserve one designated root node $r$ of the tree $T$. Note that the following proof of Span 1 does not work for general node weights. However, they work for our map weights due to a non-trivial argument using the observation stated above.  The proof requires several steps. We first prove that $T$ can be replaced with another tree whose degree is bounded by 9+6$\eps$. Next we prove that the tree can be truncated and replaced with simpler trees having bounded number of joining nodes if its height is more than $\frac{2}{\eps}+1$. Here again the proof is complicated due to presence of weights on the nodes.

\begin{lemma}[Span $1$]
\label{lem:span1}
There is another tree $T'$ such that it 1) is rooted at $r$, 2) spans all vertices in $T \cap P$, 3) has cost at most $(1 + 2\eps)c_W(T)$, and 4) has at most $2^{\OO(1/\eps)}$ joining vertices with $P$.
\end{lemma}

\noindent\textit{Proof.} 
First we prove the intermediate claim below. 

\noindent \textit{Claim.} There is a tree $T''$ which is rooted at $r$, each node has at most $\Delta = 9+6\eps$ children and has cost $c_W(T'') \leq (1+\eps) c_W(T)$. 

\noindent\textit{Proof of Claim.} 

\begin{wrapfigure}{R}{0.5\textwidth}
\centering
\includegraphics[width=0.48\textwidth]{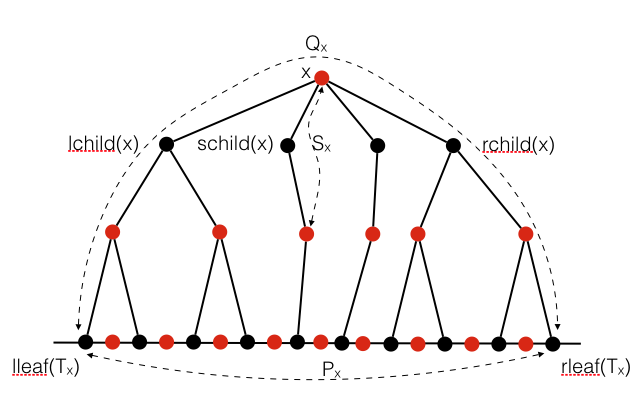}
\caption{ 
\footnotesize Illustration of definitions with respect to a node $x$ on the graph $T_x \cup P_x$. The red vertices have weight one, while the black vertices have weight zero. }
\label{fig:span1}
\end{wrapfigure}

We start with some definitions. For any node $x\in V(T)$, we use $T_x$ to denote the subtree rooted at $x$, and we refer to its leftmost child as ${\rm lchild}(x)$, and its rightmost child is referred to as ${\rm rchild}(x)$. The leftmost leaf of any tree $T$ is referred to as $\lleaf(T)$, while the rightmost leaf is referred to as $\rleaf(T)$. The subpath of $P$ from the $\lleaf(T_x)$ to $\rleaf(T_x)$ is referred to as $P_x$. If $x$ has at least 3 children, then $S_x$ is defined to be the shortest path from $x$ to $P_x$, excluding the vertex on $P_x$, which avoids ${\rm rchild}(x)$, and ${\rm lchild}(x)$.  Let ${\rm schild}(x)$ denote the child of $x$ on $S_x$. 
Let $Q_x$ be the path from $\lleaf(T_x)$ to $\rleaf(T_x)$ which passes through $x$, as shown in Figure~\ref{fig:span1}.

Assume $x$ has at least $\Delta +1$ children. We show that the tree $T''_x$ consisting of $S_x\cup P_x$ has weight at most $(1+\eps)c_W(T_x)$. We first consider the case when $w(x)=0$. Due to $\eps$-shortness, we have $c_W(P_x) \leq (1+\eps) c_W(Q_x) $. As $x$ has cost zero, it does not contribute anything extra when we add the path $S_x$ to $P_x$, i.e. $c_W(T''_x) = c_W(S_x) + c_W(P_x) = c_W(S_x) + (1+\eps)c_W(Q_x) \leq (1+\eps)c_W(T_x) + w(x)$. Therefore, $T''_x$ has the claimed cost.

We next consider the case when $w(x)=1$. Clearly, if $T \setminus (P_x\cup Q_x\cup S_x)$ has a node $v$ with $w(v)=1$, then
\begin{align*}
c_W(T''_x) =& c_W(S_x) + c_W(P_x)       \\
\leq & c_W(S_x) + (1+\eps) c_W(Q_x)\\
\leq & (1+\eps)c_W(T_x) - w(v) + w(x)\\
\leq & (1+\eps)c_W(T_x).
\end{align*}
Hence, let us assume that every path from $x$ to $P_x$ avoiding ${\rm lchild}(x)$ and ${\rm rchild}(x)$ has no weight one node on it. Due to Observation~\ref{obs:alternation_bip}, there are exactly two types of such paths, either $(a)$ $x \leftrightarrow p$, where $p\in P$ with $w(p)=0$, or $(b)$ $x \leftrightarrow v \leftrightarrow p$, where $p\in P$ with $w(p)=1$, and $w(v)=0$.
Let the number of paths of type-$a$ be $d_0$, the number of paths of type-$b$ be $d_1$, and let $d$ denote the number of children of $x$. Clearly, we have $d_0 + d_1 = d-2$.
Let $t_L$ be the $2^{nd}$ child of $x$ from the left, and let $t_R$ be the $2^{nd}$ child of $x$ from the right.
Let us denote the path in $T_x$ from 
$\lleaf(T_{t_L})$
to $\rleaf(T_{t_R)}$ by $Q'_x$.
Also, let $P'_x$ denote the subpath of $P_x$ from 
$\lleaf(T_{t_L})$
to $\rleaf(T_{t_R)}$.
As $Q'_x$ can have weight at most $3$, due to $\eps$-shortness, we have $c_W(P'_x) \leq 3\cdot (1+\eps)$. Notice that due to Observation~\ref{obs:alternation_bip}, we have $d_0 - 1 \leq c_W(P'_x) $, combining it with $d_1 \leq c_W(P'_x)$, we get $d-3 \leq 6+6\eps$. Therefore, the number of children of $x$ is at most $9+6\eps$.

 To get the tree as claimed earlier, we start at the root vertex $r$ of $T$, processing the subtrees recursively, and replace any vertex $v$ with more than $\Delta$ children with $T''_v$. Therefore every vertex after the replacement has bounded degree. This finishes the proof of the claim.
 \qed
 
Let $T''$ be the tree obtained in the last paragraph. We first get rid of degree two vertices from the graph. For any vertex $v \in T''$ with parent $p$ and only one child $c$, replace the path $p \leftrightarrow v \leftrightarrow c$ by the edge $(p,c)$, and modify the weight of $p$ to be $w(p) := w(p) + w(v)$. On applying this rule exhaustively we end up with a tree $T'''$ in which the degree of each vertex is at least $3$, in other words each node has at least $2$ children. Next, we show that if the tree $T'''$ is `too tall', then the tree has two cheap consecutive levels, which we can always include in a solution.

Let the height of the tree $T'''$ be $\ell$, with root at level $1$. We denote the nodes on a level $i \in [\ell]$ by $X_i$, and the sum of weight of nodes in $X_i$ is denoted using $Y_i$. If $\ell \leq \frac{2}{\eps}+1$, then the tree $T'''$ has at most $\Delta^{\frac{2}{\eps}+1}$ joining nodes and the tree $T'''$ satisfies the required properties.
Otherwise, there exists some level $j^* \in [\lceil \frac{2}{\eps} \rceil]$, such that $Y_{j^*}+Y_{j^*+1} \leq \eps c_W(T''')$. This follows from the averaging argument, as we have $\sum_{i \in [\lceil \frac{2}{\eps}\rceil]} \frac{Y_i + Y_{i+1}}{\frac{2}{\eps}} \leq \frac{2c_W(T''')}{\frac{2}{\eps}} \leq \eps c_W(T''')$. 
Therefore, the total cost of the nodes in $X_{j^*} \cup X_{j^*+1}$ is very small, and we can simply include them in a solution. 
For any $v\in V(T)$, let $S'_v$ be a path from $v$ to $P_v$, excluding the vertex on $P_v$, which passes through ${\rm lchild}(v)$. Note that $S'_v$ is disjoint from $Q_v$ except at $v$, and ${\rm lchild}(v)$.
For each $v \in X_{j^*}$, we replace $T'''_v$ with the tree ${\bot}_v =  P_v \cup S'_v$. Let the tree obtained from the operation above be $T'''_{\bot}$, we have 
\begin{align*}
c_W(T'''_{\bot}) = & \sum_{i \leq {j^*}-1} Y_i + \sum_{v \in X_{j^*}} c_W(\bot_v) \\
= & \sum_{i \leq {j^*}-1} Y_i + \sum_{v \in X_{j^*}} (c_W(P_v)+c_W(S'_v)) \\
\leq & \sum_{i \leq {j^*}-1} Y_i + \sum_{v \in X_{j^*}} ((1+\eps)c_W(Q_v)+c_W(S'_v)) \\
\leq & \sum_{i \leq {j^*}-1} Y_i + \sum_{v \in X_{j^*}} ((1+\eps)c_W(T_v) + w(v) + w({\rm lchild}(v)) \\
\leq & (1+2\eps)c_W(T''').
\end{align*}
Therefore, $T'''_{\perp}$ satisfies the required properties. This completes the proof of the lemma.
\qed

Finally we prove the property called as Span 2, here we need to preserve two designated vertices in the resulting simplified tree. We do not provide a proof as the following lemma follows in an analogous manner to the edge weighted case.

\begin{lemma}[Span 2, Lemma 10.6,~\cite{bkm09}]
\label{lem:span2}
Let $r,s \in V(T)$. There is another tree $T'$ such that it 1) is rooted at $r$, 2) spans all vertices in $\{r, s\} \cup (T \cap P)$, 3) has cost at most $(1 + c\eps)c_W(T)$, and 4) has at most $\OO(\frac{\rho_\eps}{\eps})$ joining vertices with $P$, where $\rho_\eps$ is the number of joining nodes due to application of Span $1$ in Lemma~\ref{lem:span1}.
\end{lemma}

Now we have all the ingredients to conclude the following structural property of a brick $B$. The proof structure for this structural lemma is exactly same as the proof for Theorem 10.7,~\cite{bkm09}. We do the same decomposition by just plugging-in Lemma~\ref{lem:brick-prop},~\ref{lem:portals},~\ref{lem:span0},~\ref{lem:span1} and~\ref{lem:span2} instead of Lemma 6.10, 7.1, 7.2, 10.2, 10.4 and 10.6 respectively.

\begin{lemma}[Brick-Structural Theorem 10.7,~\cite{bkm09}]
\label{thm:brick-structural}
Let $B$ be a map-weighted plane graph with boundary $N\cup E \cup S \cup W$ satisfying lemma~\ref{lem:brick-prop}. Let $F$ be a forest in $B$ then there is another forest $\Tilde{F}$ in $B$ satisfying the following properties:
\begin{enumerate}
    \item Any two vertices of $N \cup S$ connected in $F$ are also connected in $\Tilde{F}$.
    \item The number of joining vertices of $\Tilde{F}$ with $N \cup S$ is at most $g(\eps)$.
    \item $c(\Tilde{F}) \leq (1+c\eps)c(F)$.
\end{enumerate}
where $g(\eps) = \OO(2^{1/\eps})$\footnote{We get worse factor of $\OO(2^{1/\eps})$ because of Lemma~\ref{lem:span2}. There is a more involved argument to get a bound polynomial in $1/\eps$ which we skip for simplicity.} and $c$ is some fixed constant.
\end{lemma}

Now we are ready to sketch the proof of the spanning property.
\begin{proof}[Proof of the spanning property]
  Consider the optimum solution to the input instance $F = OPT(W, R)$. Decompose the edges of $F$ into forests as follows: for each brick $B$, let $F_B$ be the set of edges of $F$ strictly enclosed by the boundary $\partial B$ of this brick. Let $F_M$ be the remaining edges of $F$ which lie on the mortar graph $MG$.

  For brick $B$, let now $F'_B$ be the minimal forest of $F_B \cup E \cup W$. Apply Lemma~\ref{thm:brick-structural} to the brick $B$ and $F'_B$. This results in $\Tilde{F_B}$ which has a constant number of joining vertices. Now, let $D_B$, be the set of subpaths of $\partial B$ from each joining vertex to its closest portal in $P(B)$. Consider now any connected component $T$ of $\Tilde{F_B} \cup D_B$. For all such $T$, replace it with the optimum Steiner tree $T^*$ spanning the same set of portals as $T$. Let $F^*_B$ be the resulting forest which includes trees $T^*$ and additionally paths $D_B$.
  Finally, let $F^*$ be the union over all forests $F^*_B$ and $F_M$.

  By construction $F^* \subseteq H$ as it consists only of mortar edges and optimum Steiner trees. We claim that $F^*$ is a solution to the instance $W$ of cost at most $(1+\eps)OPT(W, R)$. This follows from Lemma~\ref{thm:brick-structural} in a similar way that the proof of Theorem 3.2~\cite{bkm09}. A more detailed analysis is left for the full paper.
\qed
\end{proof}

\end{document}